 \documentclass{article}
\usepackage{doi}
\usepackage{url}

 \usepackage{amsmath,amsthm,amsfonts,amssymb}
\usepackage[mathscr]{euscript}
 \usepackage{hyperref} 
\usepackage{graphicx}
\usepackage{color}
  \usepackage{epsfig}
  \usepackage{pdfsync} 
 
\usepackage{caption}
\usepackage[us,12hr]{datetime}                    %
\usepackage{fancyhdr}                                  
\begin{document}
 
\title{\bf On weak regularity requirements of the relaxation modulus in  viscoelasticity}
\date{}
\maketitle     


{\small
\begin{center}
 {\sc Sandra Carillo} \\
Dipartimento di Scienze di Base e Applicate \\
per l'Ingegneria, 
Universit\`{a}   di Roma  {\textsc{La Sapienza}},\\
 Via Antonio Scarpa 16,  00161 Rome,  
Italy \\
\&\\
I.N.F.N. - Sezione Roma1, Gr. IV - M.M.N.L.P.,  Rome, Italy \\[10pt]
 {\sc Michel Chipot} \\
IMath,  University of Z\"urich\\
Winterthurerstrasse 190, 8057, Z\"urich,  Switzerland \\
[10pt]
 {\sc Vanda Valente} \\
Istituto per le Applicazioni del Calcolo {\it
M. Picone} \\

  Via dei Taurini 19, 00185 Roma, Italy  \\[10pt]
 {\sc Giorgio Vergara Caffarelli} \\
Dipartimento di Scienze di Base e Applicate \\
per l'Ingegneria, 
Universit\`{a}   di Roma {\textsc{La Sapienza}}  \\
 Via Antonio Scarpa 16,  00161 Rome,   Italy \\[10pt]
\end{center}
}

\numberwithin{equation}{section}
\allowdisplaybreaks

\newtheorem{theorem}{Theorem}[section]
\newtheorem{lemma}[theorem]{Lemma}

 \smallskip

{\bf Abstract.}  
The existence and uniqueness of solution to a
one-dimen\-sional hyperbolic integro-differential problem arising in
viscoelasticity is here considered. The kernel, in the linear viscoelasticity equation,
represents the relaxation function which is characteristic of the considered material. 
Specifically, the case of a kernel,  which does not satisfy the {\it classical} regularity requirements is analysed. 
This choice is suggested by applications according to
the literature to model  a wider variety of materials.  
A notable example of kernel, not satisfying the classical regularity requirements, is represented by a wedge continuous function.  Indeed, the linear integro-differential viscoelasticity equation, characterised by a suitable wedge continuous relaxation function,   is shown to give the classical linear wave equation via a  limit procedure.

  \newcommand{\nor}[1]{\|#1\|}
\newcommand{\R}   {\mathbb{R}}
\newcommand{\N}   {\mathbb{N}}
\newcommand{\C}   {\mathbb{C}}
\newcommand{\Z}   {\mathbb{Z}}
\newcommand{\0}   {\mbox{\bf 0}}
\newcommand{\res}    {\mbox{\rm res}}
\newcommand{\eproof} {\framebox{}}
\newcommand{\app }[1]{\stackrel{\mbox{\scriptsize \rm(A.{#1})}}{=}}
\newcommand{\V}     {\mbox{\it V}}
\newcommand{\U}     {\mbox{\it U}}
\newcommand{\vv}     {\mbox{\it v}}
\newcommand{\uu}     {\mbox{\it u}}
\newcommand{\ww}     {\mbox{\it w}}
\newcommand{\Skw}     {\mbox{\rm Skw}\,}
\newcommand{\Sym}    {\mbox{\rm Sym}\,}
\newcommand{\skw}     {\mbox{\rm skw}\,}
\newcommand{\sym}    {\mbox{\rm sym}\,}
\newcommand{\grad}    {\mbox{\rm grad}\,} 
\newcommand{\divv}    {\mbox{\rm div}\,} 
\newcommand{\Dym}    {\mbox{\scriptsize \rm Dym}}


{\bf{AMScode:}}~{74H20, 35Q74, 45K05, 74D05}

\section{Introduction}
\setcounter{equation}{0}
 \subsection{General framework of the problem}
The mathematical model of viscoelasticity aims to take into account the behaviour 
of those materials whose mechanical behaviour is determined not only by the present 
but also on its past {\it history}.  Those materials whose mechanical and/or thermodynamical 
response depends significantly  on their past history are termed {\it materials with memory} 
according to the literature, see \cite{Fabrizio-Morro-book92, FL} where
the   physical interpretation of the model as well as its 
thermodynamical admissibility are studied. Since then  a wide variety of results and 
applications have been obtained in many different fields. Among the many, referring 
to the model of viscoelastic body, 
to give a feeling of the many 
different  areas in which this model turns out to be of interest, we mention some
 bio-mechanical applications.
In \cite{bones18}  
 synthetic tissues which mimic human bones  are investigated. Indeed, the 
viscoelastic model is applied to a variety of biologically inspired tissues; as an example, 
 in \cite{bones18}  bones, while in \cite{cardio} cardiological tissues are considered. 
A model of apples regarded as viscoelastic bodies is studied in \cite{apples, apples2018}.
%
%
 The present investigation concerns viscoelastic bodies and their mechanical behaviour
 aiming to widen the range of applicative cases the  theory can be applied to. 
 Accordingly, following the spirit of   previous investigations, 
further generalisations of the  classical model, suggested by applicative examples, 
and  not previously  considered, are addressed to.
 
%
The present investigation is part of a   research project whose interest is
focussed on materials with memory and, on one side their mechanical or thermodynamical 
characterisation \cite{BUMI, [82]}, on the other one to study of assigned initial boundary 
value problems  to  answer to the question of existence and, possibly, uniqueness of the 
admitted solution.
Under this viewpoint, the growing applicative interest on new materials, both
artificial, such as polymers, and/or materials of biological origin, requires to
adapt and, possibly, consider less restrictive assumptions when the mathematical
model is constructed. Accordingly, when a
materials characterized by mechanical response which is  modeled by a 
Volterra type integro-differential equation, is studied generalized  
kernels  need to be considered.  
Thus,   in  \cite{DIE2013} a one-dimensional viscoelasticity 
   problem characterised by a singular kernel is studied and, later, extended in 
    \cite{pericolanti, S2015}. 

The novelty of the present investigation consists in new
relaxed requirements 
{imposed on}
the kernel of the integro-differential viscoelasticity problem.
Specifically, inspired by \cite{GVC1, GVC2}, relaxations functions exhibiting a jump 
discontinuity in their first derivative are considered.

 \subsection{Mathematical of the problem}

\medskip
The one-dimensional problem is investigated. Accordingly, 
the linear integro-differential equation which models the displacement response
in a one dimensional viscoelastic body, whose configuration is denoted as 
$\Omega=(a,b)\subset\R$, can be written as  
\begin{equation} \label{eql}
{u}_{tt} = G(0) {u}_{xx} 
     + \int_{0}^t \dot{G}(t-\tau)
{u}_{xx}(\tau) d\tau + f, 
\end{equation}
where  $u$ indicates the displacement while  the
history of the material as well as an external force, if present,  are  included in the
term   $f$. The initial boundary value problem is assigned when the following initial 
and boundary conditions
\begin{equation} \label{ibcu}
 {u}\vert_{t=0}=u_0(x),\quad {u}_t\vert_{t=0}
={u}_1(x),
\quad {u}\vert_{\partial\Omega\times(0,T)}=0~~, ~~ t< T
\end{equation}
are imposed. The classical model (see for instance \cite{D,D1}), when a one-dimensional
 isotropic and homogeneous viscoelastic body is considered, prescribes that  the kernel $G$ 
 satisfies the conditions:
\begin{equation}
\dot{{G}} \in L^1(\R^+)~~,~~{{G}}(t) = {{G}}_0
+ \displaystyle{\int_{0}^{t} { \dot{G}(s) ~d s}~,~{{G}}(\infty)} =
 \displaystyle{\lim_{t\to \infty}{{G}}(t)}
 \end{equation}
where ${{G}}(\infty) >0$;  the {\it fading memory property} is also enjoyed, that is
\begin{equation*}\label{9}
{~\forall\epsilon > 0~~ \exists \, \tilde a = a \,(\epsilon,
E^t)\in\R^+ \textstyle{s.t.}~\forall a > \tilde a, 
\displaystyle{\left\vert\int_{0}^{\infty}\!\! { \dot G(s+a)
E^t(s) ~d s} \right\vert \! < \! \epsilon ~}~,
}
\end{equation*}
where ${ E}^t(\tau) := { E}(t-\tau)$ denotes the {\it{strain past history}} of the material.
The  weak formulation of the problem \eqref{eql}-  \eqref{ibcu} reads 
\begin{equation}\label{step2} 
 \displaystyle u({t})=\int_0^t {K}(t-\tau) u_{xx}(\tau) d\tau + u_1 t +u_0
 +\int_0^t d \tau\int_0^\tau f(\xi) d\xi\,,
\end{equation}
where $K$ is  termed {\it integrated relaxation function} since it denotes the integral 
of the relaxation function  defined as
\begin{equation}\label{K} 
 \displaystyle K(\xi):= \int_0^\xi {G}(\tau)  d\tau ~.
\end{equation}
Now,  the relaxation function $G(t)$ is assumed to satisfy the relaxed conditions 
\begin{equation}\label{G}
\forall~ t\in(0,\infty), ~~~G(t)>0,\quad   G 
~~ \text{is non-increasing and convex.}
\end{equation}
In particular, we impose that  the relaxation function $ {G}$ is { nonincreasing and convex} 
since these conditions are induced by the adopted physical model 
\cite{Fabrizio-Morro-book92, GG}.
In the present study, the choice to adopt  (\ref{step2}) as the model equation to study 
   the viscoelasticity problem allows to impose weaker  requirements  on 
   the relaxation function $ {G}$ aiming to  model a wider class of materials.
Thus, one of the novelties is that the function $ {G}$ is assumed  to be 
\begin{equation}\label{Gt}
G \in C^0[0,T] ~~ \forall T\in \R~,
\end{equation}
 {but no requirement  is imposed on its time derivatives.}
Note that, consequently the results are {referring to} 
a  non regular relaxation functions, 
which may exhibit jump discontinuities in the  time derivative $\dot G$, not defined on the 
whole interval $(0,T)$ according to an example,  presented in \cite{GVC1},    where the initial 
value $G(0)$ of the relaxation function is finite. This generalisation is introduced to describe  
viscoelastic materials which are not included in previous studies in \cite{DIE2013} 
wherein the relaxation functions is 
regular, i.e.  $G \in L^1(0,T)\cap C^2(0,T)$, $ \forall T\in \R$, but may be unbounded  at 
$t=0$. 
 \subsection{A notable example}  
A special case is considered in this Subsection to stress the interest of the case under investigation. Let
\begin{equation}\label{G1}
G(s)=\left\{\begin{array}{l}
{G_{\infty}- G_0 \over a} s +G_0 ~~~~~ 0\le s\le a\\
G_{\infty}~~~~~~~~~\qquad~~~~~~~ s\ge a 
\end{array}
\right.\end{equation}
\smallskip
as depicted in the following Fig. 1.
\begin{figure}[h]
\centering
\includegraphics[width=0.70\textwidth,angle=0] {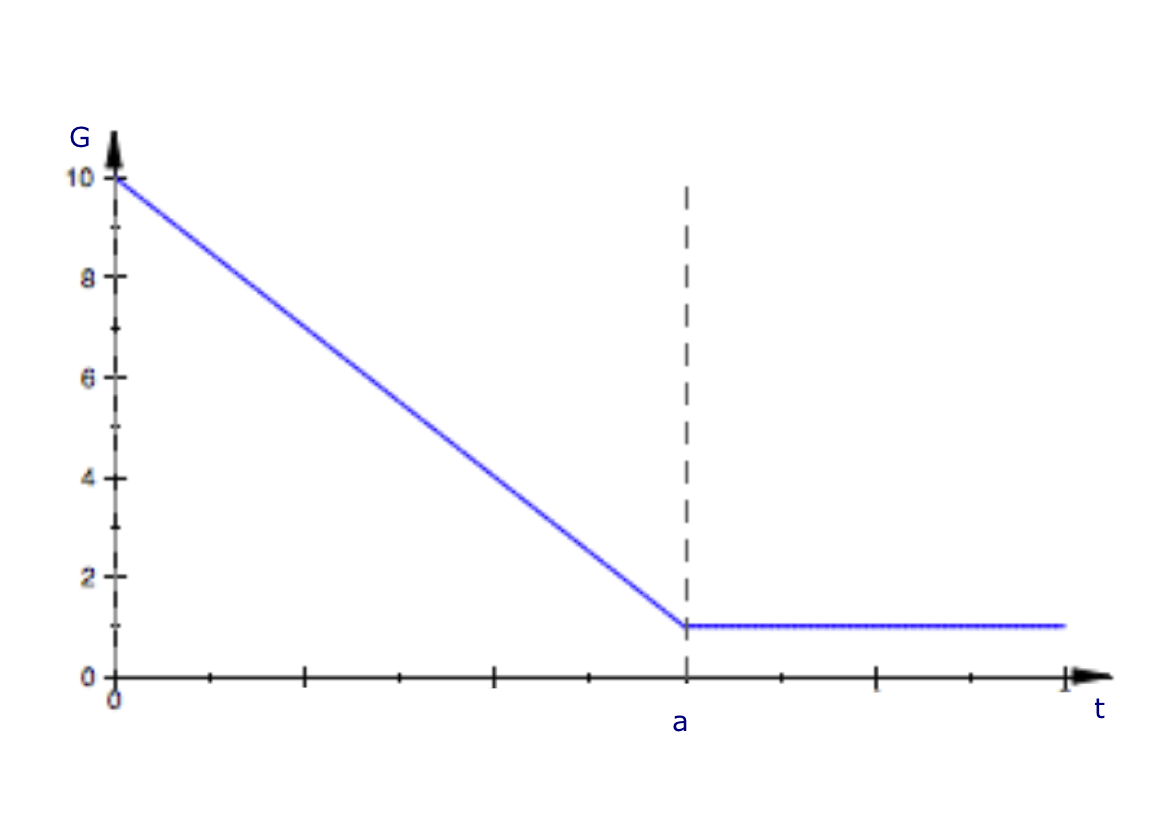}
\caption{\small Relaxation function} 
\end{figure}
The evolution equation (\ref{eql})
\begin{equation} \label{eql-ex}
{u}_{tt} = G(0) {u}_{xx} 
     + \int_{0}^t \dot{G}(t-\tau)
{u}_{xx}(\tau) d\tau + f, 
\end{equation}
now, on substitution of the given relaxation function $G$,  reads
\begin{equation} \label{eql-ex2}
{u}_{tt} = G(0) {u}_{xx} 
     + \left[ G_{\infty}- G_0\right]  \int_{\tilde t}^t {1 \over a} 
{u}_{xx}(s) ds + f, ~~~ \text{where} ~~~ \tilde t:=\max \{0, t-a\}~~,
\end{equation}
that is, when $a \to 0$, $t-a>0$ definitely. Hence, on application of the 
mean value theorem, the limit $a \to 0$ gives
\begin{equation} \label{eql-ex3}
{u}_{tt} = G(0) {u}_{xx} 
     +\left[ G_{\infty}- G_0\right]  {u}_{xx}  + f, 
     \end{equation}
i.e.
\begin{equation} \label{eql-ex4}
{u}_{tt} =      G_{\infty}      {u}_{xx}  + f, 
     \end{equation}
where $G_0$ does not appear since  the two terms cancel out and, hence, 
the linear wave equation is obtained. This example shows the relation between
linear wave equation and the linear viscoelasticity equation.
Specifically, when we consider the wedge continuous 
relaxation function in the example, the linear viscoelasticity equation admits 
the linear wave equation as its limit  wherein the speed is $c^2= G_{\infty}$.

\subsection{Content outline}
The outline of the article is as follows. In Section 2, a suitable  {\it mollified} formulation of 
the relaxation function is introduced, then corresponding  approximated regular problems, 
which depend on a small parameter $\varepsilon\in\R^+$,
are considered. The two key ingredients are, on one side, the link between the assigned
initial boundary value  problem \eqref{eql} -  \eqref{ibcu} and its integral weak formulation
 \eqref{step2}. On the other side, we prove that the limit, as $\varepsilon\to 0$, of the 
 sequence of solutions admitted by the  approximated regular problems coincides with
 the solution of the non regular initial boundary value  problem under investigation.
Suitable {\it a priori} estimates are proved: the proof of existence of the weak solution 
admitted by the problem under investigation, and constructed via  limit process, is based 
on them.  

The subsequent  Sections 3 and 4, are devoted, in turn, to show  the existence, 
and its uniqueness,  of the weak solution admitted by the problem.  Sections 3 is concerned 
about the limit of the sequence of solutions to the approximated problems.
A key role is played by the estimate in Section 2. 
The proof of the  uniqueness of the weak solution provided in Section 3 is given in  Section 4. 

\noindent
Perspectives and open problems are mentioned in the closing Section 5.

\section{Approximated problems }

This Section opens with the introduction of a mollified  relaxation function which depends on
a small parameter $\varepsilon\in\R^+$. Corresponding to each value of the parameter 
$\varepsilon$, a regular kernel problem is written. Thus, a sequence of {\it approximated} 
regular problems is constructed. On application of a Lemma, proved in \cite{DIE2013}, 
suitable {\it a priori} estimates are obtained. 
First of all, we introduce
 \begin{equation}\label{1}
G^{\varepsilon}(t):=\int_{t-\varepsilon}^{t+\varepsilon}{\rho 
\left({{t-\tau}\over \varepsilon}\right)~ {1\over \varepsilon} G( \varepsilon+\tau) d \tau}
\end{equation}
where $\rho:\R\to \R^+$ is a $C^\infty$-function with compact support in the unit intervall (-1,1), such that its integral is equal to one, that is
\begin{equation}\label{1.0}
\int_{\R}{\rho 
\left({{t-\tau}\over \varepsilon}\right)~ {1\over \varepsilon}  d \tau}=1;
\end{equation}
 so that \eqref{1} can be rewritten as
\begin{equation}
G^{\varepsilon}(t):=\int_{\R}{\rho 
\left({{t-\tau}\over \varepsilon}\right)~ {1\over \varepsilon} G( \varepsilon+\tau) d \tau}.
\end{equation}
Note that by \eqref{1} one has also 
 \begin{equation}\label{1.1}
G^{\varepsilon}(t):=\int_{-\varepsilon}^{\varepsilon}{\rho 
\left({{\tau}\over \varepsilon}\right)~ {1\over \varepsilon} G( \varepsilon+ t-\tau) d \tau},
\end{equation}
so that if $G$ is nonnegative, nonincreasing and convex so is $G^{\varepsilon}$.

Now, consider the approximated problem $P^{\varepsilon}$ given by the following
integro-differential equation
\begin{equation} \label{eql-ep}
 {u}^{\varepsilon}_{tt} = G^{\varepsilon}(0) {u}^{\varepsilon}_{xx} 
     + \int_{0}^t \dot{G}^{\varepsilon}(t-\tau)
{u}^{\varepsilon}_{xx}(\tau) d\tau + f,~~0<\varepsilon\ll 1,
\end{equation}
where $G^{\varepsilon}$ is given by \eqref{1} and, hence, $G^{\varepsilon}(0)$ as well as $ \dot{G}^{\varepsilon}$ are well defined, together with associated initial and boundary conditions
\begin{equation} \label{ibcu-ep}
 {u}^{\varepsilon}\vert_{t=0}=u_0,\quad {u}^{\varepsilon}_t\vert_{t=0}
={u}_1(x),
\quad {u}^{\varepsilon}\vert_{\partial\Omega\times(0,T)}=0~~, ~~ t< T.
\end{equation}
Notably, the problem $P^{\varepsilon}$  is 
regular and, hence, admits a unique solution \cite{D,D1}. In addition,  the corresponding
weak formulation is:
\begin{equation}\label{step2ep} 
 \displaystyle u^{\varepsilon}({t})=\int_0^t {K}^{\varepsilon}(t-\tau) u^{\varepsilon}_{xx}(\tau) d\tau + u_1 t +u_0
 +\int_0^t d \tau\int_0^\tau f(\xi) d\xi\,,
\end{equation}
where
\begin{equation}\label{Kep} 
 \displaystyle K^{\varepsilon}(\xi):= \int_0^\xi {G}^{\varepsilon}(\tau)  d\tau ~.
\end{equation}
The next Section is devoted to consider the limit, as ${\varepsilon}\to 0$ of
such an approximated integral problem to the integral problem (\ref{step2}). 

The next aim of this Section is to establish an estimate which allows to 
prove the needed convergence of the approximated solutions as 
${\varepsilon}$ goes to zero.

The following Lemma 2.1 
provides the needed estimates   to prove the
 Theorem, provided in the  subsequent Section, which  allows us to establish
 the existence result we aimed to.

As proved  in \cite{D}, \cite{D1}, the regular  linear problem 
(\ref{eql-ep}) - (\ref{ibcu-ep}), where $G^{\varepsilon}$ is given in \eqref{1},  admits a unique solution.
In particular, the following Lemma can be stated.

\smallskip

\noindent
{\bf Lemma 2.1}
 \label{2.1}  { \it
 Let ${u}^{\varepsilon}$ denote the unique solution admitted by the problem 
{\rm (\ref{eql-ep})--(\ref{ibcu-ep})}, then it follows
\begin{align}
 \label{form1} 
 &
  {1\over{2}}  {d\over{dt}} \int_{\Omega}{G}^{\varepsilon}(t)\vert{u}^
{\varepsilon}_x\vert^2 dx
-{1\over{2}} {d\over{dt}} \int_{0}^t ds \int_{\Omega} \dot{G}^{\varepsilon}(s)
\vert{u}^{\varepsilon}_x(t)-{u}^{\varepsilon}_x(t-s)\vert^2 dx +
\notag
\\
&\qquad\qquad
 + {1\over{2}}
{d\over{dt}}\int_{\Omega}\vert{u}^{\varepsilon}_t\vert^2 dx   
= \int_{\Omega} f\,{u}^{\varepsilon}_t dx +
 {1\over{2}} \int_{\Omega} \dot{ {G}^{\varepsilon}}(t) \vert{u}^{\varepsilon}_x\vert^2 
dx +\\
&
\notag
-{1\over{2}} \int_{0}^t ds \int_{\Omega} \ddot{ {G}^{\varepsilon}}(s)
\vert{u}^{\varepsilon}_x(t)-{u}^{\varepsilon}_x(t-s)\vert^2  dx~. 
\end{align} 

}
\rm

\smallskip

\noindent {\bf Proof.} 
The proof follows the lines of the proof  of Lemma 2.1 in \cite{DIE2013}.

Specifically, when we recall the conditions  (\ref{G}), imposed on $G$, and we consider 
$G^{\varepsilon}$, they imply $\dot G^{\varepsilon}<0$ and  $\ddot G^{\varepsilon}>0$,  in the 
integration  time interval $(0,t)$. Hence, these conditions, combined with integration over the
 time interval $(0,t)$, leads to
\begin{align}
\label{ineq1} 
&
 \displaystyle{{1\over{2}}   \int_{\Omega} G^{\varepsilon}(t)\, \vert{u}_x\vert^2\, dx + 
{1\over{2}} \int_{\Omega} \vert{u}_t\vert^2\, dx} 
\\ 
&
\notag
 \displaystyle{\le \int_{\Omega} \int_0^t f\,{u}_t\, dx \,ds + {1\over{2}}
\int_{\Omega} G^{\varepsilon}(0)\, \vert{u}_x(0)\vert^2\, dx + {1\over{2}} \int_{\Omega} 
\vert{u}_1\vert^2\, dx}\,.
\end{align}
The latter, when the initial data are taken into account, implies
\begin{equation}
{1\over{2}} \int_{\Omega} G^{\varepsilon}(t)\, \vert{u}_x\vert^2\, dx + {1\over{2}}
\int_{\Omega} \vert{u}_t\vert^2\, dx - \int_0^t \int_{\Omega}
\vert{u}_t\vert^2\, dx ds\le C(f, {u}_1, u_0, G_0). 
\end{equation}
Hence, on application of Gronwall's lemma,
\begin{equation}
\displaystyle{{1\over{2}} \int_{\Omega} G^{\varepsilon}(t)\, \vert{u}_x\vert^2\, dx + 
{1\over{2}}
\int_{\Omega} \vert{u}_t\vert^2\, dx \le e^T C(f, {u}_1, u_0, G_0).}
\end{equation}
Since   $G$ is nonincreasing, choosing $2\varepsilon \leq 1$, $t\leq T$ and using \eqref{1.0} one has 
\begin{equation}
G^{\varepsilon}(t)\geq\int_{t-\varepsilon}^{t+\varepsilon}{\rho 
\left({{t-\tau}\over \varepsilon}\right)~ {1\over \varepsilon} G( 2\varepsilon+t) d \tau} \geq G(1+T).
\end{equation}
Thus it holds
\begin{equation}\label{2.10}
\displaystyle{ {1\over{2}} \int_{\Omega}  \vert{u}_x\vert^2\, dx + {1\over{2}}
\int_{\Omega} \vert{u}_t\vert^2\, dx \le \alpha e^T
C(f, {u}_1, u_0)}, 
\end{equation}
wherein $\alpha= \max\{(G(T+1))^{-1}, 1\}.$
This last estimate  is required  later on.

\section{Integral Problem: Solution Existence}
The aim of this Section is to prove that, when we let ${\varepsilon} \to 0$, there exists
a weak limit, say  $u$, of any sequence  of solutions $u^{\varepsilon}$ of   
the approximated problem,  $P^{\varepsilon}$  \eqref{step2ep}.  
That is, consider any sequence of functions $u^{\varepsilon}$, solutions to  \eqref{step2ep}, 
it weakly converges to  $u$, the solution to (\ref{eql}) -- (\ref{ibcu}).
The key fact is the regularity of the problem $P^{\varepsilon}$ which implies the existence and
uniqueness of its  solution $u^{\varepsilon}$.  
Then, on use of the estimate (\ref{2.10}), it follows  that there exists
 a subsequence $\{ {\varepsilon_h}\}, h\in \N$ such that there exists a convergent subsequence 
of solutions $\{ u^{\varepsilon_h}\}$  
\begin{align}
\label{24}
&
u^{{\varepsilon_h}}\longrightarrow~~ u~~ \textrm{weakly ~in} 
~ H^1(0,T, H^1_0(\Omega)) ~~ \textrm{as} ~ {{\varepsilon_h} \to 0};
\\
&
\label{25}
u^{{\varepsilon_h}}\longrightarrow~~ u~~ \textrm{strongly ~in} ~ L^2(D)  
~~ \textrm{as} ~ {{\varepsilon_h} \to 0};
\end{align}
hence,
\begin{equation}\label{cclim} 
 {\displaystyle{\exists \,\, u(t)=\lim_{{\varepsilon_h} \to 0} {u^{{\varepsilon_h}}}({t})}\,
~~ \textrm{in} ~~L^2\left((a,b)\times(0,T)\right),}
\end{equation}
where $D=(a,b)\times(0,T)$ and $u^{{\varepsilon_h}}$ is solution to the problem  (\ref{eql-ep}) - (\ref{ibcu-ep}). Thus,  the following result can be  proved.

\begin{theorem}  The  integral problem (\ref{step2}) admits a weak solution represented by the function $ {\displaystyle{u(t)=\lim_{{\varepsilon_h} \to 0} {u^{{\varepsilon_h}}}({t})}}$.

 \end{theorem}

{\begin{proof}[\bf Proof.]
Consider the approximated problem (\ref{eql-ep}) - (\ref{ibcu-ep}), 
{{which}} admits the unique solution 
$ u^{{\varepsilon_h}}$; the associated weak formulation reads
\begin{equation}\label{ccpb2-eps} 
 \displaystyle{{P^{{\varepsilon_h}}:~~~~~~~~ u^{{\varepsilon_h}}({t})=\int_0^t {K}^
{{\varepsilon_h}}(t-\tau) u^{{\varepsilon_h}}_{xx}(\tau) d\tau + u_1 t + u_0
+\int_0^t d \tau
\int_0^\tau f(\xi) d\xi\, },}
\end{equation}
where 
\begin{equation}\label{Keps} 
 \displaystyle{ K^{{\varepsilon_h}}(\xi):= \int_0^\xi {G}^{{\varepsilon_h}}(\tau)  d\tau ~.}
\end{equation}
Then, introduce the {\it test 
functions} $\varphi$, which depend on both the time and space variables and 
satisfy homogeneous b.c.s at $a$ and $b$, i.e., at the {\it boundaries} of the $
\Omega\subset\R$ domain 
\begin{equation}\label{icphi}
\varphi \in C^\infty(D)~, ~~D=(a,b)\times(0,T),~~ \textstyle{s.t.}~~ \varphi(a,t)=
\varphi(b,t) =0 ~~~ \forall t\in  (0,T)~.
\end{equation}
On integration over $D$, after multiplication  by $\varphi$ of  (\ref{ccpb2-eps}), allows to write
\begin{equation}\label{weak-eps} 
\begin{array}{cl@{\hspace{0.5ex}}c@{\hspace{1.0ex}}l}   
 \displaystyle{ \int\!\!\! \int_D u^{{\varepsilon_h}}({t})\,\varphi \,dx \,dt =
 \int\!\!\! \int_D  \varphi\left\{\int_0^t {K}^{{\varepsilon_h}}(t-\tau) u^{{\varepsilon_h}}_{xx} (\tau) d\tau +  u_1 t \right.}\\ \displaystyle{\qquad\qquad\qquad\qquad\qquad\qquad\qquad\qquad
\left.+ u_0 
+\int_0^t d \tau\int_0^\tau f(\xi) d\xi\, \right\} dx dt .}
\end{array}
\end{equation}
First of all, note that all the terms 
\begin{equation}\label{t1} \displaystyle{ 
 \int\!\!\! \int_D  \varphi\left\{ u_1 t + u_0  +\int_0^t d \tau\int_0^\tau f(\xi) d\xi\, \right
\} dx dt }
\end{equation}
depend only on the initial data and/or on the {\it history}  of the  material 
with memory, supposed  known.  Since all these terms are assumed  regular, then the 
boundedness of $D$, implies the integral over $D$, in (\ref{t1}) is well defined and finite. 
Furthermore,  in (\ref{t1}), ${{\varepsilon_h}}$ does not appear,
hence, all the terms therein are unchanged in the limit ${{\varepsilon_h}}\to 0$.

Accordingly, 
the quantity to focuss the attention on is the following
\begin{equation}\label{pbint-eps} 
 \displaystyle{ 
 \int\!\!\! \int_D \varphi \, dx dt  \int_0^t {K}^{{\varepsilon_h}}(t-\tau) u^{{\varepsilon_h}}_{xx}
 (\tau) d \tau  }.
\end{equation}
On use of the homogeneous boundary conditions  (\ref{icphi}) satisfied by the  test functions 
$\varphi$,    integration by parts with respect to  the space variable, two times, gives : 
\begin{equation}\label{pbint-eps2} 
 \displaystyle{ 
 \int\!\!\!\! \int_D\! \varphi \, dx dt \int_0^t \! {K}^{{\varepsilon_h}}(t-\tau) u^{{\varepsilon_h}}_{xx}
 (\tau) d \tau =   \int\!\!\!\! \int_D \! \varphi \, dx dt {d\over{dx}} \int_0^t\! {K}^{{\varepsilon_h}}
 (t-\tau) u_x^{{\varepsilon_h}} (\tau) d\tau }.
\end{equation}
\begin{equation*}
= - \displaystyle{ 
 \int\!\!\!\! \int_D\! \varphi_x \, dx dt \int_0^t \! {K}^{{\varepsilon_h}}(t-\tau) u^{{\varepsilon_h}}_{x}
 (\tau) d \tau }.
\end{equation*}
The weak convergence of $u^{{\varepsilon_h}}_{x}$ to $u_x$ combined with the pointwise convergence of ${K}^{{\varepsilon_h}}$ to $K$, proved as follows, imply the existence of the weak solution. Indeed, by definition 
\begin{equation}
{K}^{\varepsilon_h}(t)- {K}(t)= \int_0^t {G}^{{\varepsilon_h}}(s)  ds  -\int_0^t {G}(s)  ds=
\end{equation}
\begin{equation*}
 \int_0^t ds \left\{ \int_{s-\varepsilon}^{s+\varepsilon}{\rho 
\left({{s-\tau}\over \varepsilon}\right)~ {1\over \varepsilon} G( \varepsilon+\tau) d \tau} 
- {G}(s) \right\}=
\end{equation*}
\begin{equation*}
 \int_0^t ds \int_{s-\varepsilon}^{s+\varepsilon} d \tau\left\{{\rho 
\left({{s-\tau}\over \varepsilon}\right)~ {1\over \varepsilon} \left[G( \varepsilon+\tau)- {G}(s)\right] }  \right\} \to 0, \text{as}~ \varepsilon \to 0 .
\end{equation*}
since $\lim_{\varepsilon \to 0} \left[G( \varepsilon+\tau)- {G}(s)\right] =0$.
\end{proof}

\section{Integral Problem: Solution Uniqueness}
In this Section the solution uniqueness is considered. 

\begin{theorem}
{Given  the integral problem  {\rm{(\ref{step2})}}, i.e.,
\begin{equation}\label{pb2} 
{ \displaystyle u({t})=\int_0^t {K}(t-\tau) u_{xx}(\tau) d\tau + u_1 t 
+\int_0^t d 
\tau\int_0^\tau f(\xi) d\xi\,}
\end{equation} 
it admits a  unique weak solution.}
\end{theorem} 

 \begin{proof}[\bf Proof.]
Assume (\ref{pb2}) admits two different solutions, say $v$ and $\tilde v$, 
then the linearity  implies that also any linear combination of them is again 
a solution to (\ref{pb2}). Thus, consider $w:=v-\tilde v$; it turns out to solve
\begin{equation}\label{uniqueness} 
 \displaystyle w({t})=\int_0^t {K}(t-\tau) w_{xx}(\tau) d\tau \,
 \end{equation}
 subject to homogeneous initial and boundary conditions, by definition.
 
 Let us denote by $w^i,\lambda_i$ the eigenfunctions and eigenvalues for the Dirichlet problem in $\Omega$ i.e. 
 satisfying
 \begin{equation}\label{uniqueness1}
  -w^i_{xx} = \lambda_i w^i \text{ in } \Omega,~~w^i =0 \text{ on } \partial\Omega.
   \end{equation}
 (Without lost of generality we can assume the eigenfunctions normalised).
 From \eqref{uniqueness}, multiplying by $w^i$ and integrating on $\Omega$ one deduces

 \begin{equation}\label{uniqueness2} 
 \displaystyle \int_\Omega w({t})w^i ~dx = - \int_0^t {K}(t-\tau)  \int_\Omega w_{x}(\tau) w^i_x ~dxd\tau \,
 \end{equation}
from which we derive by \eqref{uniqueness1} and denoting the $L^2$-scalar product by $(~,~)$

$$
(w({t}),w^i) = -\lambda_i \int_0^t {K}(t-\tau)  (w(\tau), w^i) ~d\tau .
$$
Taking absolute values it follows since $K$ is nondecreasing, nonnegative
$$
|(w({t}),w^i)| \leq K(T) \int_0^t |(w(\tau), w^i)| ~d\tau .
$$
So that by Gronwall's Lemma  $ (w({t}),w^i)=0~\forall i $ and thus $w=0$.
\end{proof}

\section{Perspectives and open problems}

The question which remains to answer to is whether it is possible to extend 
the obtained  results to the case when a magneto-viscoelastic body is considered. 
Indeed, the case of magneto-viscoelastic materials is very interesting both under the 
applicative viewpoint as well as under the mathematical one. 
Specifically, the parabolic equation that describes the evolution of the
magnetisation spin is modified by the introduction of a
further term when  the deformation effect is taken into account. There are many
publications on the subject, among them we quote the modelling of
magneto-elastic interactions described in \cite{[4]}, and the
variational analysis and some aspects on the modelling and
numerics considered in  \cite{[7],[13]}. When the mechanical response
also accounts of the past deformation, the mathematical problem
becomes an evolution problem with memory described by nonlinear
integro-differential equations. In the framework of the model proposed  in \cite{[3]},
the problem, when the two effects elastic and magnetic effects are coupled is studied 
in \cite{CSVVumi, CSVV, VV}. One of the  aspects is the change from parabolic to hyperbolic of the system, studied also in \cite{DeAngelis18}. 
The case of a magneto-viscoelastic material  where equations which couple the two effects 
of viscoelasticity and magnetisation are given, in \cite{CVV1, CVV2}.
 In  \cite{CVV1} we prove the existence and uniqueness of the solution admitted by 
  a $1$-dimensional  magneto-viscoelastic  problem. The existence result is
   generalised to the $3$-dimensional case in  \cite{CVV2} and to the case of 
   a singular kernel $1$-dimensional  magneto-viscoelastic problem in \cite{MGVS17}. 
   Singular kernel problems are studied also in \cite{[83], S2015}.
Currently under investigation is the coupling of the magnetic effects with  the 
viscoelasticity model considered in the present study.

\section*{\bf Acknowledgments.}
The authors gratefully acknowledge the 
 partial financial support of the Italian 
GNFM-INdAM and of  
Universit\`a di Roma {\textsc{La Sapienza}}. S. Carillo thanks also INFN-Sez. Roma1, Gr. MMNLP.
 M. Chipot thanks the Italian 
GNFM-INdAM and Dipartimento di Scienze di Base e Applicate
per l'Ingegneria, Universit\`{a}   di Roma  {\textsc{La Sapienza}},
for the kind hospitality.

\nocite{x}

\end{document}